\documentclass[12pt, draftclsnofoot, onecolumn]{IEEEtran}
\normalsize

\usepackage{amsfonts}
\usepackage{amsmath}
\usepackage{amssymb}
\usepackage{graphicx}
\usepackage{mathtools}
\usepackage{adjustbox}
\usepackage[colorlinks,linkcolor=black,anchorcolor=black,citecolor=black]{hyperref}
\usepackage{cite} 

\usepackage{amsthm} 

\usepackage{soul}  

\newcommand\mycom[2]{\genfrac{}{}{0pt}{}{#1}{#2}} 

\usepackage{xcolor} 

\newcommand{\R}[1]{\textcolor{red}{#1}}

\usepackage[normalem]{ulem}
\newtheorem{theorem}{Theorem}
\newtheorem{proposition}{Proposition}
\newtheorem{lemma}{Lemma}
\newtheorem{example}{Example}

\begin{document}
\title{Polarization-Adjusted Convolutional (PAC) Codes as a Concatenation of Inner Cyclic and Outer Polar- and Reed-Muller-like Codes}

\author{Mohsen Moradi\textsuperscript{\href{https://orcid.org/0000-0001-7026-0682}{\includegraphics[scale=0.06]{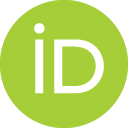}}},~\IEEEmembership{Member,~IEEE}
\thanks{This paper is based on \cite[Ch.~6]{moradi2022performance}. The author is with the Department of Electrical-Electronics Engineering, Bilkent University, Ankara TR-06800, Turkey (e-mail: moradi@ee.bilkent.edu.tr).}
}

\maketitle
\begin{abstract}
Polarization-adjusted convolutional (PAC) codes are a new family of linear block codes that can perform close to the theoretical bounds in the short block-length regime. 
These codes combine polar coding and convolutional coding. 
In this study, we show that PAC codes are equivalent to a new class of codes consisting of inner cyclic codes and outer polar- and Reed-Muller-like codes.
We leverage the properties of cyclic codes to establish that PAC codes outperform polar- and Reed-Muller-like codes in terms of minimum distance.
\end{abstract}


\section{Introduction}

\IEEEPARstart{T}{he} recently proposed polarization-adjusted convolutional (PAC) codes are a family of linear codes that can approach theoretical limits for short block lengths \cite{arikan2019sequential}.
In this paper, we establish a relationship between PAC codes and cyclic codes and demonstrate that PAC codes can be viewed as a concatenation of inner cyclic shift codes and outer polar codes.
We also investigate the weight distribution of the PAC codes.
By benefiting from cyclic shift matrices, we give a new proof to a result of \cite{li2019pre}, which says that the minimum distance, $d_{\text{min}}$, for PAC codes, is greater than or equal to the $d_{\text{min}}$ for polar- and Reed-Muller (RM)- like codes.

The usage of cyclic codes in our work is motivated by the work of \cite{luo2018analysis}, which designs and analyzes a particular permutation set of polar codes based on a $N/4$ -cyclic shift for practical applications, where $N$ is the length of the code block \cite{luo2020patent}. 
In our study, we generalize this algebraic result to the $m$-cyclic shift for $1\leq m \leq N$, offer an explicit proof, and demonstrate how the findings may be applied to be used in the PAC codes. 
In \cite{li2019pre}, they also proved that the sum of $\mathbf{g}_i$ ($i$th row of $F^{\otimes n}$ for $1\leq i<N$, where $F^{\otimes n}$ is the $n$th Kronecker power of $F = \begin{bsmallmatrix} 1 & 0\\ 1 & 1 \end{bsmallmatrix}$ with $n = \log_2 N$) with some rows below it (we represent this by $\underline{\mathbf{g}}_{i}$) has a weight greater than or equal to the weight of $\mathbf{g}_i$. 
We generalize this and prove that the sum of an odd number of clockwise cyclic shifts of $\underline{\mathbf{g}}_{i}$ also has a weight greater than or equal to the weight of $\mathbf{g}_i$.
Also, we prove that the summation of a row of the matrix $F^{\otimes n}$ with a row below it is equal to some clockwise cyclic shifts of that row, and we use this to prove that the $d_{\text{min}}$ for PAC codes is greater than or equal to $d_{\text{min}}$ for the polar- and RM-like codes.
The weight distribution of linear codes dictates the performance of maximum likelihood (ML) decoding, which can be well estimated by the union bound, particularly at high signal-to-noise ratios (SNRs). 
This implies that PAC codes outperform polar- and RM-like codes regarding error correction performance. 
The following is a brief overview of the contributions made by this study.
\begin{itemize}
    \item We proposed a new class of coding scheme equivalent to the PAC codes, consisting of cyclic codes as the inner code and polar- and RM-like codes as the outer code.
    \item We proved that adding an odd number of clockwise cyclic shifts to any row of the polar- and Reed-Muller-like codes generator matrix, added with some rows below it, does not reduce the row's weight.
    \item Using the significant inherent algebraic structure of cyclic codes, we proved that, in terms of minimum distance, PAC codes outperform polar- and RM-like codes.
    \item We proved that the summation of a row of the polar- and RM-like generator matrices with any row below it is equal to some clockwise cyclic shifts of that row.
\end{itemize}

In this work, the vectors are indicated in bold text, and the operations are performed over a binary field $\mathbb{F}_2$.
For vector $\mathbf{c} = (c_1, c_2, \cdots, c_N)$, the notation $\mathbf{c}^i$ is used to express $(c_1, \cdots, c_i)$, and the notation $\mathbf{c}_i^j$ is used to express $(c_i, \cdots, c_j)$.
For a vector $\mathbf{c}$, $(\mathbf{c})_i$ denotes its $i$th element $c_i$.
The polynomial representation of the vector $\mathbf{c} = (c_1, c_2, \cdots, c_N)$ is denoted by $\mathbf{c}(x) = c_{1}+ c_{2}x+ \cdots+ c_{N}x^{N-1}$.

The remainder of this paper is organized as follows.
Section \ref{sec: PAC coding} briefly reviews PAC codes.
Section \ref{sec: weight dsbn} proves that PAC codes can be seen as a concatenation of inner cyclic shift codes and outer polar- and RM-like codes.
Finally, Section \ref{sec: conclusion} concludes this paper.

\section{PAC Coding}\label{sec: PAC coding}

Fig. \ref{fig: PAC_encoder} shows the block diagram of the PAC encoding scheme.
A PAC code is specified by the parameters $(N, K, \mathcal{A}, T)$, where $N$ ($N=2^n$ and $n\in \mathbb{N}$) is the block length of the code word, $K$ is the length of the data word, $\mathcal{A}$ is known as the data index set, and $T$ is an upper triangular Toeplitz matrix. 
The data index set $\mathcal{A}$ defines the rate profile (data insertion) module of the PAC code that maps the data word $\mathbf{d}$ of length $K$ into a data carrier vector $\mathbf{v}$ of length $N$ s.t. $\mathbf{v}_{\mathcal{A}} = \mathbf{d}$ and $\mathbf{v}_{\mathcal{A}^c} = \mathbf{0}$, which is the input to the convolutional encoder defined by the matrix $T$, where

\begin{equation*}
T = 
\begin{bmatrix}
 c_0    & c_1    &  c_2   & \cdots & c_m    & 0      & \cdots & 0      \\
 0      & c_0    & c_1    & c_2    & \cdots & c_m    &        & \vdots \\
 0      & 0      & c_0    & c_1    & \ddots & \cdots & c_m    & \vdots \\
 \vdots & 0      & \ddots & \ddots & \ddots & \ddots &        & \vdots \\
 \vdots &        & \ddots & \ddots & \ddots & \ddots & \ddots & \vdots \\
 \vdots &        &        & \ddots & 0      & c_0    & c_1    & c_2    \\
 \vdots &        &        &        & 0      & 0      & c_0    & c_1    \\
 \vdots & \cdots & \cdots & \cdots & \cdots & 0      & 0      & c_0    
\end{bmatrix},
\end{equation*}
and can be described by a connection polynomial $\mathbf{c}(x) = c_mx^{m}+ \cdots + c_1x + c_0$ s.t. $c_0 = c_m = 1$.
For PAC codes, it is required to construct the set $\mathcal{A}$ based on the polarized cutoff rates to have a low complexity tree search algorithm \cite{moradi2021monte, moradi2023application}.  
The output of the convolutional encoder is obtained as $\mathbf{u} = \mathbf{v}T$.
Then, a one-rate polar transform (polar mapper) is applied to the output of the convolutional encoder as $\mathbf{x} = \mathbf{u}F^{\otimes n}$ to obtain the PAC codeword. 
Note that we employ a separate block for the rate profile, and the selection of the information bits can be based on polar codes, RM codes, or any other arbitrary rate profile (dubbed polar- and RM-like codes).

\begin{figure}[htbp]
\centering
	\includegraphics [width = .8\columnwidth]{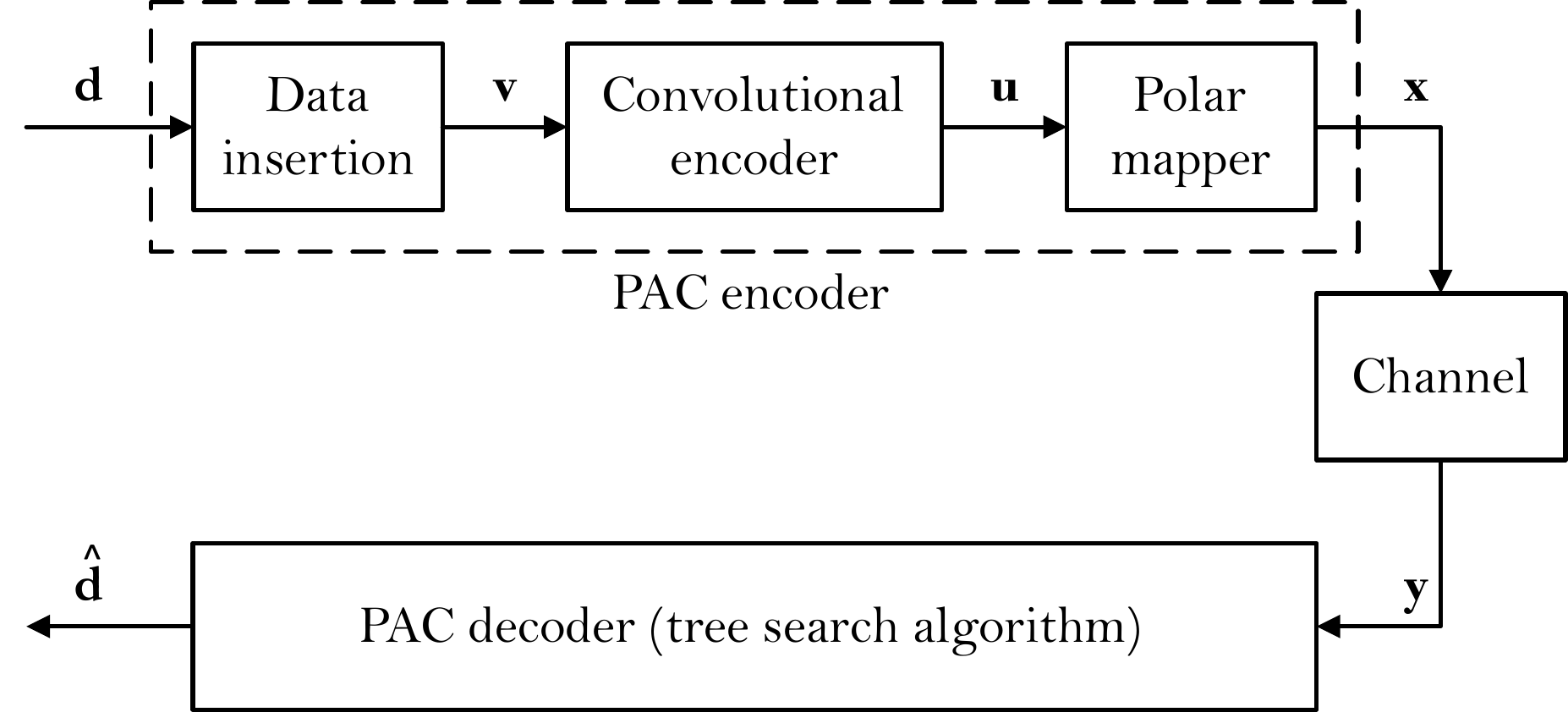}
	\caption{PAC coding scheme.} 
	\label{fig: PAC_encoder}
\end{figure}

On the receiver side, a tree search algorithm produces an estimate $\hat{\mathbf{v}}$ of $\mathbf{v}$ with the help of a metric calculator \cite{moradi2021sequential}.
Finally, an estimate of the data word is extracted from $\hat{\mathbf{v}}$. 

\begin{example}
Consider a PAC(8,3) code where the code blocklength $N$ is 8 and the length of the data bits $K$ is 3.
The polar mapper matrix is
\[
F^{\otimes 3} = \setlength{\arraycolsep}{9pt}
\begin{bmatrix}
1 & 0 & 0 & 0 & 0 & 0 & 0 & 0 \\
1 & 1 & 0 & 0 & 0 & 0 & 0 & 0 \\
1 & 0 & 1 & 0 & 0 & 0 & 0 & 0 \\
1 & 1 & 1 & 1 & 0 & 0 & 0 & 0 \\
1 & 0 & 0 & 0 & 1 & 0 & 0 & 0 \\
1 & 1 & 0 & 0 & 1 & 1 & 0 & 0 \\
1 & 0 & 1 & 0 & 1 & 0 & 1 & 0 \\
1 & 1 & 1 & 1 & 1 & 1 & 1 & 1 \\
\end{bmatrix}.
\]
Let the connection polynomial be $\mathbf{c}(x) = x^6 + x^5 + x^3 + 1$ (151 in octal form).
So, the convolutional encoder matrix is
\[
T = \setlength{\arraycolsep}{9pt}
\begin{bmatrix}
1 & 1 & 0 & 1 & 0 & 0 & 1 & 0 \\
0 & 1 & 1 & 0 & 1 & 0 & 0 & 1 \\
0 & 0 & 1 & 1 & 0 & 1 & 0 & 0 \\
0 & 0 & 0 & 1 & 1 & 0 & 1 & 0 \\
0 & 0 & 0 & 0 & 1 & 1 & 0 & 1 \\
0 & 0 & 0 & 0 & 0 & 1 & 1 & 0 \\
0 & 0 & 0 & 0 & 0 & 0 & 1 & 1 \\
0 & 0 & 0 & 0 & 0 & 0 & 0 & 1 \\
\end{bmatrix}.
\]
Assume that the data index set $\mathcal{A} = \{4,7,8 \}$.
The data insertion block inserts the data vector $\mathbf{d} = (d_1,d_2,d_3)$ into the vector $\mathbf{v}=(v_1, v_2, \cdots , v_8)$ s.t. $\mathbf{v}= (0, 0, 0, d_1, 0, 0, d_2, d_3)$.
The convolutional encoder block generates vector $\mathbf{u}$ from the vector $\mathbf{v}$ as $\mathbf{u}= \mathbf{v}T$.
Finally, the polar transform $\mathbf{x} = \mathbf{u}F^{\otimes 3}$ is computed to finish the encoding process.
\end{example}

\section{PAC Codes v. Polar Codes}\label{sec: weight dsbn}

Assume $\mathbf{g}_i$ be the $i$th row of the matrix $G_n \triangleq F^{\otimes n}$ for $1\leq i \leq N$.
The weight of $\mathbf{g}_i$ is $w(\mathbf{g}_i) = 2^{\sum_{m=1}^{n}i_m}$, where $i_{n}i_{n-1}\cdots i_{1} = \sum_{m =1}^{n}i_{m}2^{m-1}$ is the bit-index representation of index $i-1$ \cite{arikan2009channel}.
Let $C_{N}^{j}$ be the $N\times N$ clockwise cyclic shift matrix at $j$ places st. 
$\mathbf{g}_i C_{N}^{j}= (g_{i,1}, g_{i,2}, \cdots, g_{i,N})C_{N}^{j} = (g_{i,N-j+1}, g_{i,N-j+2}, \cdots g_{i,N}, \cdots, g_{i,N-j})$.
In the polynomial representation, this is equivalent to
\begin{equation}
    x^{j}\mathbf{g_i}(x) = 
    g_{i,N-j+1}+ g_{i,N-j+2}x + \cdots g_{i,N}x^{j-1}+ \cdots+ g_{i,N-j}x^{N-1}~~~ \text{mod}~~x^N-1,
\end{equation}
where $\mathbf{g_i}(x) = g_{i,1}+ g_{i,2}x+ \cdots+ g_{i,N}x^{N-1}$.

\begin{example}
Let $N =8$, $j = 3$, and $i = 6$.

\begin{align*}
C_{8}^{3} = \setlength{\arraycolsep}{9pt}
\begin{bmatrix}
    0 & 0 & 0 & 1 & 0 & 0 & 0 & 0 \\
    0 & 0 & 0 & 0 & 1 & 0 & 0 & 0 \\
    0 & 0 & 0 & 0 & 0 & 1 & 0 & 0 \\
    0 & 0 & 0 & 0 & 0 & 0 & 1 & 0 \\
    0 & 0 & 0 & 0 & 0 & 0 & 0 & 1 \\
    1 & 0 & 0 & 0 & 0 & 0 & 0 & 0 \\
    0 & 1 & 0 & 0 & 0 & 0 & 0 & 0 \\
    0 & 0 & 1 & 0 & 0 & 0 & 0 & 0 \\
\end{bmatrix},
\quad
G_{3} = 
\begin{bmatrix}
    1 & 0 & 0 & 0 & 0 & 0 & 0 & 0 \\
    1 & 1 & 0 & 0 & 0 & 0 & 0 & 0 \\
    1 & 0 & 1 & 0 & 0 & 0 & 0 & 0 \\
    1 & 1 & 1 & 1 & 0 & 0 & 0 & 0 \\
    1 & 0 & 0 & 0 & 1 & 0 & 0 & 0 \\
    1 & 1 & 0 & 0 & 1 & 1 & 0 & 0 \\
    1 & 0 & 1 & 0 & 1 & 0 & 1 & 0 \\
    1 & 1 & 1 & 1 & 1 & 1 & 1 & 1 \\
\end{bmatrix}.
\end{align*}

For $\mathbf{g}_{6} = (1,1,0,0,1,1,0,0)$,
$$\mathbf{g}_{6}C_{8}^{3} = (1,0,0,1,1,0,0,1).$$

\end{example}

For any nonzero binary vector $\mathbf{u}=(0,\cdots,0,1,u_{i+1},u_{i+2},u_{N})$, where the first $1$ is in its $i$th location,
$\mathbf{u}G_{n}= \mathbf{g}_{i}\oplus u_{i+1}\mathbf{g}_{i+1}+\cdots\oplus u_{N}\mathbf{g}_{N}$.
We use the $\underline{\mathbf{g}}_{i}$ notation to show the sum of vector $\mathbf{g}_{i}$ with some other specified rows below it.

In the following proposition, we prove that adding an odd number of clockwise cyclic shifts of any row of the matrix $G_{n}$ that is added with some rows below it (i.e., $\underline{\mathbf{g}}_{i}$ for the $i$th row of the matrix $G_{n}$) cannot decrease the weight of that row. 
In our proof, we partition the generator matrix $G_{n+1}$ into upper and lower parts as 
$G_{n+1} = 
\begin{bmatrix} 
    H_{1}\\H_{2} 
\end{bmatrix}$
s.t. the rows of $H_{1}$ and $H_{2}$ can be represented as $(\mathbf{h}_{i},\mathbf{0})$ and $(\mathbf{h}_{k},\mathbf{h}_{k})$, respectively, where both $\mathbf{h}_{i}$ and $\mathbf{h}_{k}$ are rows of matrix $G_n$.\\

\begin{proposition}\label{pro: 1}
For any vector $(u_{i+1},u_{i+2},\cdots,u_{N})\in \{0,1\}^{N-i}$ and any odd number $l = 1, 3, 5, \cdots$,
\begin{equation}
    w\left(\underline{\mathbf{g}}_{i}\sum C_{N}^l\right) \geq
    w(\mathbf{g}_{i}),
\end{equation}
where $1\leq i \leq N$,
$\underline{\mathbf{g}}_{i} = \mathbf{g}_{i}\oplus u_{i+1}\mathbf{g}_{i+1}+\cdots\oplus u_{N}\mathbf{g}_{N}$,
 and in mod $N$, there are odd numbers of distinct cyclic shifts.
\end{proposition}

\begin{proof}
The proof is by induction on $n$. 
For $n = 1$, the statement obviously holds.
By induction, suppose that the statement holds for $n$. 
We want to show that for any row $\mathbf{g}_{i}$ of $G_{n+1}$,
\begin{equation}
    w\left(\underline{\mathbf{g}}_{i}\sum C_{2N}^l\right) \geq
    w(\mathbf{g}_{i}),
\end{equation}
where $1\leq i \leq 2N$.
We divide the proof into two cases based on whether the row index $i$ is bigger than $N$.

\textbf{Case 1:} $i > N$. 

In this case, we can write the the vector $\underline{\mathbf{g}}_{i}$ as $(\underline{\mathbf{h}}_{i-N},\underline{\mathbf{h}}_{i-N})$, where $\mathbf{h}_{i-N}$ is the $(i-N)$th row of matrix $G_n$.
Hence,
\begin{equation}
\begin{split}
    &w\left(\underline{\mathbf{g}}_{i}\sum C_{2N}^l\right)
    = w\left( (\underline{\mathbf{h}}_{i-N},\underline{\mathbf{h}}_{i-N})\sum C_{2N}^l \right) \\
    & =w\left(
    \underline{\mathbf{h}}_{i-N}\sum C_{N}^l
    ,\underline{\mathbf{h}}_{i-N}\sum C_{N}^l
    \right) \\
    &= 2w\left(
    \underline{\mathbf{h}}_{i-N}\sum C_{N}^l
    \right) \geq
    2w(\mathbf{h}_{i-N})\\
    &= w(\mathbf{h}_{i-N},\mathbf{h}_{i-N}) = w(\mathbf{g}_{i}),
\end{split}
\end{equation}
where the inequality is by induction.
Notice that in the third equality, it is just possible for an even number of shifts to get equal ($l$ and $l+N$), and again the inequality will hold.
\\

\textbf{Case 2:} $i \leq N$.

In this case, the $i$th row of the matrix $G_{n+1}$ can be written as $\mathbf{g}_{i} = (\mathbf{h},\mathbf{0})$ s.t. $\mathbf{h}$ is the $i$th row of matrix $G_n$.
We have
\begin{equation}
    \underline{\mathbf{g}}_{i} =
    (\underline{\mathbf{h}},\mathbf{0}) \oplus
    (\underline{\mathbf{y}},
    \underline{\mathbf{y}}),
\end{equation}
where $(\mathbf{y},\mathbf{y})$ is a zero vector or is the $j$th row of matrix $G_{n+1}$ s.t. $j > N$.
We have
\begin{equation}
\begin{split}
    \underline{\mathbf{g}}_{i}\sum C_{2N}^l 
    &=
    (\underline{\mathbf{h}},\mathbf{0})\sum C_{2N}^l
    \oplus
    (\underline{\mathbf{y}},
    \underline{\mathbf{y}})
    \sum C_{2N}^l\\
    &=
    (\underline{\mathbf{h}},\mathbf{0})\sum C_{2N}^l
    \oplus
    \left(\underline{\mathbf{y}}\sum C_{N}^l,
    \underline{\mathbf{y}}\sum C_{N}^l\right).
\end{split}
\end{equation}
Suppose $\gamma \triangleq \left(\underline{\mathbf{y}}\sum C_{N}^l \right)_{k}$.
Also, note that the $k$th bit
\begin{equation}
    \left(\underline{\mathbf{h}}\sum C_{N}^l
    \right)_{k} =
    \underbrace{
    \left( 
    (\underline{\mathbf{h}},\mathbf{0})\sum C_{2N}^l
    \right)_{k}
    }_{\alpha}
    \oplus
    \underbrace{
    \left( 
    (\underline{\mathbf{h}},\mathbf{0})\sum C_{2N}^l
    \right)_{k+N}
    }_{\beta}
\end{equation}
is equal to 1 if 
$(\alpha, \beta) = (1,0)$ or $(\alpha, \beta) = (0,1)$ and we can have one of the following four cases:

If $\gamma = 0$ and $(\alpha, \beta) = (0,1)$, then 
$\left(\underline{\mathbf{g}}_{i}\sum C_{2N}^l \right)_{k+N} =1.$

If $\gamma = 1$ and $(\alpha, \beta) = (0,1)$, then 
$\left(\underline{\mathbf{g}}_{i}\sum C_{2N}^l \right)_{k} =1.$

If $\gamma = 0$ and $(\alpha, \beta) = (1,0)$, then 
$\left(\underline{\mathbf{g}}_{i}\sum C_{2N}^l \right)_{k} =1.$

If $\gamma = 1$ and $(\alpha, \beta) = (1,0)$, then 
$\left(\underline{\mathbf{g}}_{i}\sum C_{2N}^l \right)_{k+N} =1.$

\noindent So, 
\begin{equation}
    w(\underline{\mathbf{g}}_{i}\sum C_{2N}^l) \geq w(\underline{\mathbf{h}}\sum C_{N}^l) \geq
    w(\mathbf{h}) = w(\mathbf{g}_{i}),
\end{equation}
where the last inequality is by induction.
\end{proof}

The following lemma is useful in proving the Theorem \ref{theorem}.
\begin{lemma}\label{lemma identityShift}
For $1\leq m \leq N-1$ and $1\leq k \leq N$,
if $k+m > N$,
\begin{equation}
    \mathbf{g}_{k} = \mathbf{g}_{k}
    \sum_{\mycom{1\leq l \leq N-1~\text{s.t.}}{g_{m+1,\R{l+1}}=1}}
    C_{N}^{l},
\end{equation}
where $\mathbf{g}_{k}$ and $\mathbf{g}_{m+1}$ are the $k$th and $(m+1)$th rows of the matrix $G_{n}$.
\end{lemma}

\begin{example}
For $n = 3$, $k = 6$, and $m = 3$, we have $m + k =9 > N$. 
So, for $\mathbf{g}_{k} = (1,1,0,0,1,1,0,0)$
and
$\mathbf{g}_{m+1} = (1,\R{1},\R{1},\R{1},0,0,0,0)$
we have $\mathbf{g}_{k} = \mathbf{g}_{k}(C^{\R{1}} \oplus C^{\R{2}} \oplus C^{\R{3}})$.\\
\end{example}

\begin{proof}
The proof is by induction. Suppose that Lemma \ref{lemma identityShift} is true for $n$, and we prove it for $n+1$ case.
We divide the proof into two cases based on whether $k$ exceeds $N$.

\textbf{Case 1:} $k > N$. \\
The $k$ th row of matrix $G_{n+1}$ can be written as
$\mathbf{g}_k = (\mathbf{h}_{k-N}, \mathbf{h}_{k-N})$,
where $\mathbf{h}_{k-N}$ is the $(k-N)$th row of matrix $G_{n}$.
If $k+m > 2N$, then $k+m-N > N$ and by induction we have that
\begin{equation}
    \mathbf{h}_{k-N} = \mathbf{h}_{k-N} 
    \sum_{\mycom{1\leq l \leq N-1~\text{s.t.}}{h_{m+1,l+1}=1}}
    C_{N}^{l}.
\end{equation}
From this, we have
\begin{equation}
\begin{split}
    &\mathbf{g}_k
    \sum_{\mycom{1\leq l \leq 2N-1~\text{s.t.}}{g_{m+1,l+1}=1}}
    C_{2N}^{l}\\
    & = (\mathbf{h}_{k-N}, \mathbf{h}_{k-N})
    \sum_{\mycom{1\leq l \leq 2N-1~\text{s.t.}}{g_{m+1,l+1}=1}}
    C_{2N}^{l} \\
    &= (\mathbf{h}_{k-N}
    \sum_{\mycom{1\leq l \leq N-1~\text{s.t.}}{h_{m+1,l+1}=1}}
    C_{N}^{l}, \mathbf{h}_{k-N}
    \sum_{\mycom{1\leq l \leq N-1~\text{s.t.}}{h_{m+1,l+1}=1}}
    C_{N}^{l})\\
    & = (\mathbf{h}_{k-N}, \mathbf{h}_{k-N})
    = \mathbf{g}_k.
\end{split}
\end{equation}

\textbf{Case 2:} $k \leq N$. \\
Suppose that $\mathbf{g}_{k} = (\mathbf{h}_k,\mathbf{0})$ is the $k$th row of matrix $G_{n+1}$ and $k+m > 2N$, where $\mathbf{h}_{k}$ is the $k$th row of matrix $G_{n}$.
Because $k\leq N$, we can write $m$ as $m = j+N$ s.t. $j<N$.
From $j+N+k > 2N$, we have $j+k > N$ and by induction we have
\begin{equation}
    \mathbf{h}_{k} = \mathbf{h}_{k}
    \sum_{\mycom{1\leq l \leq N-1~\text{s.t.}}{h_{j+1,l+1}=1}}
    C_{N}^{l}.
\end{equation}
Also note that for any vector $\mathbf{x}$ of length $N$ and $1\leq l \leq N - 1$ we have
\begin{equation}\label{eq: (x,0)}
    (\mathbf{x},0)(C_{2N}^{l} \oplus C_{2N}^{l+N}) = 
    (\mathbf{x}C_{N}^{l},\mathbf{x}C_{N}^{l}).
\end{equation}
As an example, for $N=8$ and $l=3$, 
\begin{equation*}
\begin{split}
&(\mathbf{x},0)(C_{2N}^{l} \oplus C_{2N}^{l+N}) \\
&= ( x_1,x_2,x_3,x_4,x_5,x_6,x_7,x_8,0,0,0,0,0,0,0,0 ) C_{16}^{3}\\
&+ ( x_1,x_2,x_3,x_4,x_5,x_6,x_7,x_8,0,0,0,0,0,0,0,0 ) C_{16}^{11}\\
&= ( 0,0,0,x_1,x_2,x_3,x_4,x_5,x_6,x_7,x_8,0,0,0,0,0 ) \\
&+ ( x_6,x_7,x_8,0,0,0,0,0,0,0,0,x_1,x_2,x_3,x_4,x_5 )\\
&= ( x_6,x_7,x_8,x_1,x_2,x_3,x_4,x_5,x_6,x_7,x_8,x_1,x_2,x_3,x_4,x_5 ) \\
&=(\mathbf{x}C_{N}^{l},\mathbf{x}C_{N}^{l}).
\end{split}
\end{equation*}

So we have
\begin{equation}
\begin{split}
    &\mathbf{g}_{k}
    \sum_{\mycom{1\leq l \leq 2N-1~\text{s.t.}}{g_{m+1,l+1}=1}}
    C_{2N}^{l}
    = (\mathbf{h}_{k},\mathbf{0})
    \sum_{\mycom{1\leq l \leq N-1~\text{s.t.}}{(h_{j+1,l+1},h_{j+1,l+1})=(1,1)}}
    C_{2N}^{l}\\
    & =(\mathbf{h}_{k}
    \sum_{\mycom{1\leq l \leq N-1~\text{s.t.}}{h_{j+1,l+1}=1}}
    C_{N}^{l},
    \mathbf{h}_{k}
    \sum_{\mycom{1\leq l \leq N-1~\text{s.t.}}{h_{j+1,l+1}=1}}
    C_{N}^{l})\\
    & =(\mathbf{h}_{k},\mathbf{h}_{k}) 
    = \mathbf{g}_{k},
\end{split}
\end{equation}
where the second equality is by \eqref{eq: (x,0)} and the third equality is by induction.
\end{proof}

Let us define an $N\times N$ upper-triangular bidiagonal matrix $D_{N}$ as
\begin{equation}
    D_{N} \triangleq 
    \begin{bmatrix}
    1      \ & 1      \ & 0         \ & \ldots      \ & 0        \\
    0      \ & 1      \ & 1         \ &             \ & \vdots   \\
    \vdots \ & 0      \ & \ddots    \ & \ddots      \ &          \\
           \ & \vdots \ &           \ &  1          \ & 1        \\
    0      \ & 0      \ & \ldots    \ &             \ & 1 
\end{bmatrix},
\end{equation}
which is an upper-triangular Toeplitz matrix that has 1 in its main diagonal and upper diagonal elements (its first row has 1 in the first and second positions), with all other entries being zero.

We also define matrix $D_{N}^{m}$ as an upper-triangular Toeplitz matrix that its first row has 1 in the first and $(m+1)$th positions, where $1\leq m \leq N-1$, and we define $D_{N}^{0}$ as the identity matrix.

The following theorem relates the matrix $D_{N}^m$ to clock-wise cyclic shift matrices $C_{N}^{l}$.
Note that the relation is trivial when $m=0$; both matrices multiplied by $G_n$ are identity matrices.

\begin{theorem}\label{theorem}
For any $m \leq N-1$,
\begin{equation}
    D_{N}^{m}G_{n} = G_{n}
    \sum_{\mycom{1\leq l \leq N-1~\text{s.t.}}{g_{m+1,l+1}=1}} C_{N}^{l}.
\end{equation}
This is equivalent to saying that if $i = k + m \leq N$, then
\begin{equation}\label{eq: theoremSum}
    \mathbf{g}_k \oplus \mathbf{g}_i = \mathbf{g}_{k}
    \sum_{\mycom{1\leq l \leq N-1~\text{s.t.}}{g_{i-k+1,l+1}=1}} C_{N}^{l},
\end{equation}
and if $k+m > N$, then
\begin{equation}\label{eq: identityShift}
    \mathbf{g}_{k} = \mathbf{g}_{k}
    \sum_{\mycom{1\leq l \leq N-1~\text{s.t.}}{g_{m+1,l+1}=1}} C_{N}^{l}
\end{equation}
\end{theorem}

\begin{example}
For $n = 2$, 
\begin{equation*}
 G_{n} = 
    \begin{bmatrix}
    1      \ & 0      \ & 0         \ & 0     \\
    1      \ & 1      \ & 0         \ & 0     \\
    1      \ & 0      \ & 1         \ & 0     \\
    1      \ & 1      \ & 1         \ & 1       
\end{bmatrix},   
D^2 = 
    \begin{bmatrix}
    1      \ & 0      \ & 1         \ & 0     \\
    0      \ & 1      \ & 0         \ & 1     \\
    0      \ & 0      \ & 1         \ & 0     \\
    0      \ & 0      \ & 0         \ & 1       
\end{bmatrix},
\end{equation*}
\begin{equation*}
    C^2 = 
    \begin{bmatrix}
    0      \ & 0      \ & 1         \ & 0     \\
    0      \ & 0      \ & 0         \ & 1     \\
    1      \ & 0      \ & 0         \ & 0     \\
    0      \ & 1      \ & 0         \ & 0       
    \end{bmatrix} ,
\end{equation*}
\begin{equation}
D^{2}G_{n} = G_{n}C^{2} = 
    \begin{bmatrix}
    0      \ & 0      \ & 1         \ & 0     \\
    0      \ & 0      \ & 1         \ & 1     \\
    1      \ & 0      \ & 1         \ & 0     \\
    1      \ & 1      \ & 1         \ & 1       
\end{bmatrix}.
\end{equation}
\end{example}

\begin{example}
For $n = 2$, 
\begin{equation*}
 G_{n} = 
    \begin{bmatrix}
    1      \ & 0      \ & 0         \ & 0     \\
    1      \ & 1      \ & 0         \ & 0     \\
    1      \ & 0      \ & 1         \ & 0     \\
    1      \ & 1      \ & 1         \ & 1       
\end{bmatrix},\    
D^3 = 
    \begin{bmatrix}
    1      \ & 0      \ & 0         \ & 1     \\
    0      \ & 1      \ & 0         \ & 0     \\
    0      \ & 0      \ & 1         \ & 0     \\
    0      \ & 0      \ & 0         \ & 1       
\end{bmatrix},
\end{equation*}
\begin{equation*}
C^1 = 
    \begin{bmatrix}
    0      \ & 1      \ & 0         \ & 0     \\
    0      \ & 0      \ & 1         \ & 0     \\
    0      \ & 0      \ & 0         \ & 1     \\
    1      \ & 0      \ & 0         \ & 0       
\end{bmatrix} ,\
C^2 = 
    \begin{bmatrix}
    0      \ & 0      \ & 1         \ & 0     \\
    0      \ & 0      \ & 0         \ & 1     \\
    1      \ & 0      \ & 0         \ & 0     \\
    0      \ & 1      \ & 0         \ & 0       
\end{bmatrix}, \
C^3 = 
    \begin{bmatrix}
    0      \ & 0      \ & 0         \ & 1     \\
    1      \ & 0      \ & 0         \ & 0     \\
    0      \ & 1      \ & 0         \ & 0     \\
    0      \ & 0      \ & 1         \ & 0       
    \end{bmatrix} ,
\end{equation*}
\begin{equation}
D^{3}G_{n} = G_{n}(C^{1} \oplus C^{2} \oplus C^{3}) = 
    \begin{bmatrix}
    0      \ & 1      \ & 1         \ & 1     \\
    1      \ & 1      \ & 0         \ & 0     \\
    1      \ & 0      \ & 1         \ & 0     \\
    1      \ & 1      \ & 1         \ & 1       
\end{bmatrix}.
\end{equation}
\end{example}

\begin{proof}
The Lemma \ref{lemma identityShift} proves the equation \eqref{eq: identityShift}, and we provide proof for the case $i =k+m \leq N$. 
The proof is by induction, and we assume that \eqref{eq: theoremSum} is true for $G_{n}$; then, we prove it for $G_{n+1}$.
In this respect, we have $i = k+m \leq 2N$.
We divide the proof into 6 cases based on the values of $i$ and $k$.
The first case is for $k > N$, the second is when $i,k \leq N$, and the other four are when $k\leq N$ and ${i} > N$.

\textbf{Case 1:} $k > N$. \\
We have 
\begin{equation}
    \mathbf{g}_{k} \oplus \mathbf{g}_{i} = 
    (\mathbf{h}_{k-N}\oplus \mathbf{h}_{i-N}, \mathbf{h}_{k-N} \oplus \mathbf{h}_{i-N}),
\end{equation}
where $\mathbf{h}_{k-N}$ and $\mathbf{h}_{i-N}$ are $(k-N)$th and $(i-N)$th rows of matrix $G_{N}$, respectively.
By induction, we have
\begin{equation}
\begin{split}
    & \mathbf{g}_{k} \oplus \mathbf{g}_{i} = (\mathbf{h}_{k-N}\oplus \mathbf{h}_{i-N}, \mathbf{h}_{k-N}\oplus \mathbf{h}_{i-N}) \\
    & = (\mathbf{h}_{k-N}
    \sum_{\mycom{1\leq l \leq N-1~\text{s.t.}}{h_{i-k+1,l+1}=1}}
    C_{N}^{l}, \mathbf{h}_{k-N}
    \sum_{\mycom{1\leq l \leq N-1~\text{s.t.}}{h_{i-k+1,l+1}=1}}
    C_{N}^{l}) \\
    & = (\mathbf{h}_{k-N},\mathbf{h}_{k-N})
    \sum_{\mycom{1\leq l \leq 2N-1~\text{s.t.}}{g_{i-k+1,l+1}=1}}
    C_{2N}^{l} \\
    &= \mathbf{g}_{k} 
    \sum_{\mycom{1\leq l \leq 2N-1~\text{s.t.}}{g_{i-k+1,l+1}=1}}
    C_{2N}^{l}.
\end{split}
\end{equation}
As an example for $N =8$, $k = 12$, and $i = 15$, the $12$th and $15$th rows of matrix $G_{4}$ are
\begin{equation*}
\begin{split}
    &\mathbf{g}_{12} = (\mathbf{h}_{4}, \mathbf{h}_{4}
    ) = (1,1,1,1,0,0,0,0,1,1,1,1,0,0,0,0),\\
    & \mathbf{g}_{15} = (\mathbf{h}_{7}, \mathbf{h}_{7}
    ) = (1,0,1,0,1,0,1,0,1,0,1,0,1,0,1,0).
\end{split}
\end{equation*}
By knowing that $i-k+1=4$, we have
\begin{equation*}
    \mathbf{g}_{12}\oplus \mathbf{g}_{15} = \mathbf{g}_{12}(C_{16}^{1}\oplus C_{16}^{2}\oplus C_{16}^{3}).
\end{equation*}


\textbf{Case 2:} $k \leq N$ and $i\leq N$.\\
The $k$th and $i$th rows of the matrix $G_{n+1}$ can be written as 
$\mathbf{g}_{k} = (\mathbf{h}_{k}, \mathbf{0})$ and 
$\mathbf{g}_{i} = (\mathbf{h}_{i}, \mathbf{0})$, 
respectively, where $\mathbf{h}_{k}$ and $\mathbf{h}_{i}$ are the corresponding rows of the matrix $G_{n}$.
By induction, we have
\begin{equation}
    \mathbf{h}_k \oplus \mathbf{h}_i = \mathbf{h}_{k}
    \sum_{\mycom{1\leq l \leq N-1~\text{s.t.}}{h_{i-k+1,l+1}=1}}
    C_{N}^{l}.
\end{equation}
Note that $\mathbf{h}_{k}$ has its last $1$ at the $k$th position and shifting based on the vector $\mathbf{h}_{i-k+1}$ can shift that last $1$ up to $i-k$ position.
As a result, the last $1$ of the vector $\mathbf{h}_{k}$ can be shifted up to the $i$th position, which is less than or equal to $N$.
Considering this, we have
\begin{equation}
\begin{split}
    &\mathbf{g}_{k} \oplus \mathbf{g}_{i} =
    (\mathbf{h}_{k},\mathbf{0}) \oplus (\mathbf{h}_{i}, \mathbf{0}) =
    (\mathbf{h}_{k}\oplus \mathbf{h}_{i}, \mathbf{0})\\
    &= (\mathbf{h}_{k}
    \sum_{\mycom{1\leq l \leq N-1~\text{s.t.}}{h_{i-k+1,l+1}=1}}
    C_{N}^{l}, \mathbf{0}) \\
    &= (\mathbf{h}_{k}, \mathbf{0})
    \sum_{\mycom{1\leq l \leq 2N-1~\text{s.t.}}{g_{i-k+1,l+1}=1}}
    C_{2N}^{l} =
    \mathbf{g}_{k}
    \sum_{\mycom{1\leq l \leq 2N-1~\text{s.t.}}{g_{i-k+1,l+1}=1}}
    C_{2N}^{l},
\end{split}
\end{equation}
where the third equality is by induction and the fourth is by the discussion above.

As an example, for $N =8$, $k = 2$, and $i = 8$, the $2$nd and $8$th rows of matrix $G_{4}$ are
\begin{equation*}
\begin{split}
    &\mathbf{g}_{2} = (\mathbf{h}_{2}, \mathbf{0}
    ) = (1,1,0,0,0,0,0,0,0,0,0,0,0,0,0,0),\\
    & \mathbf{g}_{8} = (\mathbf{h}_{8}, \mathbf{0}
    ) = (1,1,1,1,1,1,1,1,0,0,0,0,0,0,0,0).
\end{split}
\end{equation*}
By knowing that $i-k+1 = 7$, we have
\begin{equation*}
    \mathbf{g}_{2}\oplus \mathbf{g}_{8} = \mathbf{g}_{2}(C_{16}^{2}\oplus C_{16}^{4}\oplus C_{16}^{6}).
\end{equation*}


\textbf{Case 3:} $k \leq N$, $i= j+N$ and $j\geq k$.\\
The $k$th, $j$th, and the $i$th rows of the matrix $G_{n+1}$ are as
$\mathbf{g}_{k}=(\mathbf{h}_{k},\mathbf{0})$, 
$\mathbf{g}_{j}=(\mathbf{h}_{j},\mathbf{0})$,
and
$\mathbf{g}_{i}=(\mathbf{h}_{j},\mathbf{h}_{j})$,
respectively, where $\mathbf{h}_{k}$ and $\mathbf{h}_{j}$ are the corresponding rows of the matrix $G_n$.
By induction, in shifting based on the $r = j-k+1$ row of matrix $G_n$, we have
\begin{equation}\label{eq: shifting}
    \mathbf{h}_{k} \oplus \mathbf{h}_{j} =
    \mathbf{h}_{k}
    \sum_{\mycom{1\leq l \leq N-1~\text{s.t.}}{h_{j-k+1,l+1}=1}}
    C_{N}^{l}.
\end{equation}
Note that the shifts are by the positions greater than or equal to 2 of the vector $\mathbf{h}_r$, which we show as 
$(\times,h_{r,2},h_{r,3},\cdots,h_{r,N})$.
In general, we say that shifting vector $\mathbf{h}_{k}$ based on vector $(\times,h_{r,2},h_{r,3},\cdots,h_{r,N})$ results in $\mathbf{h}_{k} \oplus \mathbf{h}_{j}$ as an alternative way to \eqref{eq: shifting}.

Note that from case 2, we know that with shifting $\mathbf{h}_{k}$ by row $j-k+1$, the last $1$ element of $\mathbf{h}_{k}$ will be shifted at most to the $N$th position. 
From this we can say that shifting $(\mathbf{h}_{k}, \mathbf{0})$ based on $(\times,h_{r,2},h_{r,3},\cdots,h_{r,N}, \mathbf{0})$ is $(\mathbf{h}_{k}\oplus \mathbf{h}_{j}, \mathbf{0})$.
Similarly, shifting the $k$th row $(\mathbf{h}_{k}, \mathbf{0})$ based on 
$(\mathbf{0}, \times,h_{r,2},h_{r,3},\cdots,h_{r,N})$ is
$(\mathbf{0}, \mathbf{h}_{k}\oplus \mathbf{h}_{j})$.
Also, note that shifting a vector $(x_{1},x_{2}, \cdots x_{N}, \mathbf{0})$ based on a vector that only has 1 at its $(N+1)$th position is $(\mathbf{0}, x_{1},x_{2}, \cdots x_{N})$, i.e.
\begin{equation}
    (x_{1},x_{2}, \cdots x_{N}, \mathbf{0})C_{2N}^{N} = (\mathbf{0}, x_{1},x_{2}, \cdots x_{N}).
\end{equation}
As a results, shifting the $k$th row $(\mathbf{h}_{k}, \mathbf{0})$ based on vector 
$$(\times,h_{r,2},h_{r,3},\cdots,h_{r,N},1 ,h_{r,2},h_{r,3},\cdots,h_{r,N})$$
results in
$(\mathbf{h}_{k}\oplus \mathbf{h}_{j}, \mathbf{h}_{k}\oplus (\mathbf{h}_{k}\oplus \mathbf{h}_{j})) = (\mathbf{h}_{k}\oplus \mathbf{h}_{j}, \mathbf{h}_{j})$.
This means that
\begin{equation}
\begin{split}
    & \mathbf{g}_{k}
    \sum_{\mycom{1\leq l \leq 2N-1~\text{s.t.}}{g_{i-k+1,l+1}=1}}
    C_{2N}^{l} \\
    &=
    (\mathbf{h}_{k},\mathbf{0})
    \sum_{\mycom{1\leq l \leq 2N-1~\text{s.t.}}{g_{i-k+1,l+1}=1}}
    C_{2N}^{l} \\
    &= 
    (\mathbf{h}_{k} \oplus \mathbf{h}_{j}, \mathbf{0})
    \oplus (\mathbf{0}, \mathbf{h}_{k})
    \oplus (\mathbf{0}, \mathbf{h}_{k} \oplus \mathbf{h}_{j})
    \\
    & = (\mathbf{h}_{k} \oplus \mathbf{h}_{j}, \mathbf{h}_{j}) = \mathbf{g}_{k} \oplus \mathbf{g}_{i}.
\end{split}
\end{equation}

As an example, for $N=8$, $k = 7$, and $i = 16$, the $7$th and the $16$th rows of matrix $G_{4}$ are
\begin{equation*}
\begin{split}
    &\mathbf{g}_{7}~ = (\mathbf{h}_{7}, \mathbf{0})~ = (1,0,1,0,1,0,1,0,0,0,0,0,0,0,0,0),\\
    & \mathbf{g}_{16} = (\mathbf{h}_{8}, \mathbf{h}_{8}
    ) = (1,1,1,1,1,1,1,1,1,1,1,1,1,1,1,1).
\end{split}
\end{equation*}
As $j = 8$ and $j-k+1 = 2$, the shift is based on 
$(\times,h_{2,2},\cdots,h_{2,8},1,h_{2,2},\cdots,h_{2,8})= (\times,1,0,0,0,0,0,0,1,1,0,0,0,0,0,0)$.
So
\begin{equation*}
    \mathbf{g}_{7} \oplus \mathbf{g}_{16} = 
    \mathbf{g}_{7}(C_{16}^{1}\oplus C_{16}^{8}\oplus C_{16}^{9}).
\end{equation*}


\textbf{Case 4:} $k \leq N$, $i= j+N$, $j\leq k$, and $j > N/2$.\\
We have $N/2 < k \leq N$ and $N+N/2 < i$.
The $k$th and $i$th rows of the matrix $G_{n+1}$ can be written as 
$\mathbf{g}_{k} = (\mathbf{h}_{k}, \mathbf{0}) = (\mathbf{x}, \mathbf{x}, \mathbf{0}, \mathbf{0})$ and 
$\mathbf{g}_{i} = (\mathbf{h}_{j}, \mathbf{h}_{j}) = (\mathbf{y},\mathbf{y},\mathbf{y},\mathbf{y})$, 
respectively, where $\mathbf{h}_{k}$ and $\mathbf{h}_{j}$ are the corresponding rows of the matrix $G_{N}$, and $\mathbf{x}$ and $\mathbf{y}$ are corresponding rows of the matrix $G_{N/2}$.

We have
\begin{equation}
    \mathbf{h}_{k-N/2} \oplus \mathbf{h}_{j} = 
    (\mathbf{x}, \mathbf{0}) \oplus (\mathbf{y}, \mathbf{y}) = 
    (\mathbf{x} \oplus \mathbf{y}, \mathbf{y}),
\end{equation}
and by induction
\begin{equation}\label{eq: base Case4}
    (\mathbf{x}\oplus \mathbf{y}, \mathbf{y}) = 
    (\mathbf{x}, \mathbf{0})
    \sum_{\mycom{1\leq l \leq N-1~\text{s.t.}}{h_{j-k+N/2+1,l+1}=1}}
    C_{N}^{l}.
\end{equation}
Notice that $j-k+N/2+1 \leq N/2$.
So
$\mathbf{g}_{j-k+N/2+1} = (\mathbf{r}, \mathbf{0},\mathbf{0},\mathbf{0})$,
where $\mathbf{r}$ is a row of matrix $G_{N/2}$.
In this way, $i-k+1 = (j-k+N/2+1) + N/2$ and
$\mathbf{g}_{i-k+1} = (\mathbf{r}, \mathbf{r},\mathbf{0},\mathbf{0})$.

From \eqref{eq: base Case4} we can see that shifting 
$(\mathbf{x},\mathbf{0},\mathbf{0},\mathbf{0})$ 
by 
$(\times,r_2,\cdots,r_{N/2},\mathbf{0},\mathbf{0},\mathbf{0})$ 
is 
$(\mathbf{x}\oplus\mathbf{y},\mathbf{y}, \mathbf{0},\mathbf{0})$.

Moreover, shifting $(\mathbf{x},\mathbf{0},\mathbf{0},\mathbf{0})$
by
$(\mathbf{0},1,r_2,\cdots,r_{N/2},\mathbf{0},\mathbf{0})$ 
is
$(\mathbf{0}, \mathbf{x}\oplus(\mathbf{x}\oplus \mathbf{y}),\mathbf{y},\mathbf{0}) = (\mathbf{0}, \mathbf{y},\mathbf{y},\mathbf{0})$.

Likewise, shifting $(\mathbf{0},\mathbf{x},\mathbf{0},\mathbf{0})$
by 
$(\times,r_2,\cdots,r_{N/2},\mathbf{0},\mathbf{0},\mathbf{0})$ 
is also
$(\mathbf{0},\mathbf{x}\oplus \mathbf{y},\mathbf{y}, \mathbf{0})$.

Finally, shifting 
$(\mathbf{0},\mathbf{x},\mathbf{0},\mathbf{0})$
by
$(\mathbf{0},1,r_2,\cdots,r_{N/2},\mathbf{0},\mathbf{0})$ 
is
$(\mathbf{0},\mathbf{0}, \mathbf{x}\oplus (\mathbf{x}\oplus \mathbf{y}),\mathbf{y}) = (\mathbf{0},\mathbf{0}, \mathbf{y},\mathbf{y})$.

By considering these four shifting together, the shifting of 
$(\mathbf{x},\mathbf{x},\mathbf{0},\mathbf{0})$
by
$$(\times,r_2,\cdots,r_{N/2},1,r_2,\cdots,r_{N/2},\mathbf{0},\mathbf{0})$$
is
$(\mathbf{x}\oplus \mathbf{y},\mathbf{y}, \mathbf{0},\mathbf{0}) \oplus
(\mathbf{0}, \mathbf{y},\mathbf{y},\mathbf{0}) \oplus
(\mathbf{0},\mathbf{x}\oplus \mathbf{y},\mathbf{y}, \mathbf{0}) \oplus
(\mathbf{0},\mathbf{0}, \mathbf{y},\mathbf{y}) =
(\mathbf{x}\oplus \mathbf{y}, \mathbf{x}\oplus \mathbf{y},\mathbf{y},\mathbf{y}).$
So, we conclude that
\begin{equation}
\begin{split}
    &\mathbf{g}_{k}
    \sum_{\mycom{1\leq l \leq 2N-1~\text{s.t.}}{g_{i-k+1,l+1}=1}}
    C_{2N}^{l} =
    (\mathbf{x},\mathbf{x},\mathbf{0},\mathbf{0})
    \sum_{\mycom{1\leq l \leq 2N-1~\text{s.t.}}{g_{i-k+1,l+1}=1}}
    C_{2N}^{l}\\
    & = (\mathbf{x}\oplus \mathbf{y}, \mathbf{x}\oplus \mathbf{y},\mathbf{y},\mathbf{y})\\
    & = (\mathbf{x},\mathbf{x},\mathbf{0},\mathbf{0}) \oplus 
    (\mathbf{y},\mathbf{y},\mathbf{y},\mathbf{y}) = \mathbf{g}_{k} \oplus \mathbf{g}_{i}.
\end{split}
\end{equation}

As an example, for $N =8$, $k = 8$, and $i = 15$, the $8$th and $15$th rows of matrix $G_{4}$ are
\begin{equation*}
\begin{split}
    &\mathbf{g}_{8} ~
    =(\mathbf{h}_{8}, \mathbf{0})~
    = (\mathbf{x},\mathbf{x},\mathbf{0},\mathbf{0})\\
    & ~~~~~~~~~~~~~~~~~= (1,1,1,1,1,1,1,1,0,0,0,0,0,0,0,0),\\
    & \mathbf{g}_{15} = (\mathbf{h}_{7}, \mathbf{h}_{7}) =
    (\mathbf{y},\mathbf{y},\mathbf{y},\mathbf{y})\\
    & ~~~~~~~~~~~~~~~~~= (1,0,1,0,1,0,1,0,1,0,1,0,1,0,1,0).
\end{split}
\end{equation*}
As $i-k+1 = 8$,
Shifting is based on the vector $\mathbf{g}_{i-k+1} = (\times,r_2,\cdots,r_{N/2},1,r_2,\cdots,r_{N/2},\mathbf{0},\mathbf{0}) = (\times,1,1,1,1,1,1,1,0,0,0,0,0,0,0,0)$.
So
$\mathbf{g}_8 \oplus \mathbf{g}_{15} = \mathbf{g}_{8}(C^1\oplus C^2\oplus \cdots \oplus C^7)$.

\textbf{Case 5:} $k \leq N$, $i= j+N$, $j\leq k$, and $1 \leq j \leq N/2$ and $k > N/2$.\\
We have that $N < i \leq N + N/2$ and $1\leq i-k+1\leq N$.
The $k$th and the $i$th rows of the matrix $G_{n+1}$ can be written as
$\mathbf{g}_{k} = (\mathbf{h}_{k},\mathbf{0})= (\mathbf{x},\mathbf{x},\mathbf{0},\mathbf{0})$,
and
$\mathbf{g}_{i} = (\mathbf{h}_{j},\mathbf{h}_{j})= (\mathbf{y},\mathbf{0},\mathbf{y},\mathbf{0})$,
respectively, where $\mathbf{h}_{k}$ and $\mathbf{h}_{j}$ are the corresponding rows of the matrix $G_{N}$, and $\mathbf{x}$ and $\mathbf{y}$ are corresponding rows of the matrix $G_{N/2}$.
Also $\mathbf{h}_{k-N/2} = (\mathbf{x}, \mathbf{0})$ and $\mathbf{h}_{i-N/2} = (\mathbf{y}, \mathbf{y})$.
By induction, we have that
\begin{equation}\label{eq: case5Base}
\begin{split}
    & (\mathbf{x}, \mathbf{0})
    \sum_{\mycom{1\leq l \leq N-1~\text{s.t.}}{h_{i-k+1,l+1}=1}}
    C_{N}^{l} = 
    \mathbf{h}_{k-N/2}
    \sum_{\mycom{1\leq l \leq N-1~\text{s.t.}}{g_{i-k+1,l+1}=1}}
    C_{N}^{l}\\
    &= \mathbf{h}_{k-N/2} \oplus \mathbf{h}_{i-N/2} = 
    (\mathbf{x}\oplus \mathbf{y}, \mathbf{y}).
\end{split}
\end{equation}
As $i-k+1\leq N$, $\mathbf{g}_{i-k+1}= (\mathbf{r},\mathbf{0})$.

From \eqref{eq: case5Base} we have that shifting $(\mathbf{x},\mathbf{0},\mathbf{0},\mathbf{0})$
by
$(\times,r_2,\cdots,r_N,\mathbf{0})$ 
is
$(\mathbf{x}\oplus \mathbf{y}, \mathbf{y},\mathbf{0},\mathbf{0})$.

Also shifting 
$(\mathbf{0},\mathbf{x},\mathbf{0},\mathbf{0})$
by
$(\times,r_2,\cdots,r_N,\mathbf{0})$ 
is
$(\mathbf{0},\mathbf{x}\oplus \mathbf{y}, \mathbf{y},\mathbf{0})$.

By considering these two shifting together, the shifting of 
$(\mathbf{x},\mathbf{x},\mathbf{0},\mathbf{0})$
by
$(\times,r_2,\cdots,r_N,\mathbf{0})$
is
$(\mathbf{x}\oplus \mathbf{y}, \mathbf{y}\oplus (\mathbf{x}\oplus \mathbf{y}),\mathbf{y},\mathbf{0}) = 
(\mathbf{x}\oplus \mathbf{y}, \mathbf{x},\mathbf{y},\mathbf{0})$.
So we conclude that
\begin{equation}
\begin{split}
    &\mathbf{g}_{k}
    \sum_{\mycom{1\leq l \leq 2N-1~\text{s.t.}}{g_{i-k+1,l+1}=1}}
    C_{2N}^{l} =
    (\mathbf{x},\mathbf{x},\mathbf{0},\mathbf{0})
    \sum_{\mycom{1\leq l \leq N-1~\text{s.t.}}{g_{i-k+1,l+1}=1}}
    C_{2N}^{l}\\
    & = (\mathbf{x}\oplus \mathbf{y}, \mathbf{x},\mathbf{y},\mathbf{0})\\
    & = (\mathbf{x},\mathbf{x},\mathbf{0},\mathbf{0}) \oplus 
    (\mathbf{y},\mathbf{0},\mathbf{y},\mathbf{0}) = \mathbf{g}_{k} \oplus \mathbf{g}_{i}.
\end{split}
\end{equation}


\textbf{Case 6:} $k \leq N$, $i= j+N$, $j\leq k$, and $1 \leq j \leq N/2$ and $k \leq N/2$.\\
We have that $N < i \leq N + N/2$ and $N/2 < i-k+1\leq N$.
The $k$th and the $i$th rows of the matrix $G_{n+1}$ can be written as
$\mathbf{g}_{k} = (\mathbf{h}_{k},\mathbf{0})= (\mathbf{x},\mathbf{0},\mathbf{0},\mathbf{0})$,
and
$\mathbf{g}_{i} = (\mathbf{h}_{j},\mathbf{h}_{j})= (\mathbf{y},\mathbf{0},\mathbf{y},\mathbf{0})$,
respectively, where $\mathbf{h}_{k}$ and $\mathbf{h}_{j}$ are the corresponding rows of the matrix $G_{N}$, and $\mathbf{x}$ and $\mathbf{y}$ are corresponding rows of the matrix $G_{N/2}$.
Also $\mathbf{h}_{k} = (\mathbf{x}, \mathbf{0})$ and $\mathbf{h}_{i-N/2} = (\mathbf{y}, \mathbf{y})$.
By induction, we have that
\begin{equation}\label{eq: case6Base}
\begin{split}
    & (\mathbf{x}, \mathbf{0})
    \sum_{\mycom{1\leq l \leq N-1~\text{s.t.}}{h_{i-N/2-k+1,l+1}=1}}
    C_{N}^{l} = 
    \mathbf{h}_{k}
    \sum_{\mycom{1\leq l \leq N-1~\text{s.t.}}{h_{i-N/2-k+1,l+1}=1}}
    C_{N}^{l}\\
    &= \mathbf{h}_{k} \oplus \mathbf{h}_{i-N/2} = 
    (\mathbf{x}\oplus \mathbf{y}, \mathbf{y}).
\end{split}
\end{equation}

Note that $1 < i-N/2-k+1 \leq N/2$ and $N/2 < i-k+1 \leq N$.
So
$\mathbf{g}_{i-k+1} = (\mathbf{r}, \mathbf{r},\mathbf{0},\mathbf{0})$.

From \eqref{eq: case6Base}, shifting
$(\mathbf{x},\mathbf{0},\mathbf{0},\mathbf{0})$
by
$\mathbf{g}_{i-N/2-k+1} = (\times,r_2,\cdots,r_{N/2},\mathbf{0},\mathbf{0},\mathbf{0})$
is
$(\mathbf{x}\oplus \mathbf{y}, \mathbf{y},\mathbf{0},\mathbf{0})$.

Likewise, shifting
$(\mathbf{x},\mathbf{0},\mathbf{0},\mathbf{0})$
by
$(\mathbf{0},1,r_2,\cdots,r_{N/2},\mathbf{0},\mathbf{0})$
is
$(\mathbf{0},\mathbf{x}\oplus (\mathbf{x}\oplus \mathbf{y}), \mathbf{y},\mathbf{0})$.

By considering these two shifting together, the shifting of 
$(\mathbf{x},\mathbf{0},\mathbf{0},\mathbf{0})$
by
$$(\times,r_2,\cdots,r_{N/2},1,r_2,\cdots,r_{N/2},\mathbf{0},\mathbf{0})$$
is
$(\mathbf{x}\oplus \mathbf{y},\mathbf{0},\mathbf{y},\mathbf{0})$.

So we conclude that
\begin{equation}
\begin{split}
    &\mathbf{g}_{k}
    \sum_{\mycom{1\leq l \leq 2N-1~\text{s.t.}}{g_{i-k+1,l+1}=1}}
    C_{2N}^{l} =
    (\mathbf{x},\mathbf{0},\mathbf{0},\mathbf{0})
    \sum_{\mycom{1\leq l \leq 2N-1~\text{s.t.}}{g_{i-k+1,l+1}=1}}
    C_{2N}^{l}\\
    & = (\mathbf{x}\oplus \mathbf{y}, \mathbf{0},\mathbf{y},\mathbf{0})\\
    & = (\mathbf{x},\mathbf{0},\mathbf{0},\mathbf{0}) \oplus 
    (\mathbf{y},\mathbf{0},\mathbf{y},\mathbf{0}) = \mathbf{g}_{k} \oplus \mathbf{g}_{i}.
\end{split}
\end{equation}
\end{proof}

Assume that the matrix $T_{N}$ is an $N\times N$ upper-triangular Toeplitz matrix and that the elements along each diagonal of the matrix are identical.
We can see that $T_{N} = \sum_{m} D_{N}^m$ for some values of $m$ and by Theorem \ref{theorem},
$D_{N}^{m}G_{n} = G_{n}\sum_{l} C_{N}^{l}$ for some values of $l$.

\begin{example}
Assume $N=8$. 
For the polynomial $\mathbf{c}(x) = x^6+x^4+x^3+x+1$ (133 in octal form), the convolutional encoder matrix is as

\[
T_8 = \setlength{\arraycolsep}{9pt}
\begin{bmatrix}
    1 & 0 & 1 & 1 & 0 & 1 & 1 & 0 \\
    0 & 1 & 0 & 1 & 1 & 0 & 1 & 1 \\
    0 & 0 & 1 & 0 & 1 & 1 & 0 & 1 \\
    0 & 0 & 0 & 1 & 0 & 1 & 1 & 0 \\
    0 & 0 & 0 & 0 & 1 & 0 & 1 & 1 \\
    0 & 0 & 0 & 0 & 0 & 1 & 0 & 1 \\
    0 & 0 & 0 & 0 & 0 & 0 & 1 & 0 \\
    0 & 0 & 0 & 0 & 0 & 0 & 0 & 1 \\
\end{bmatrix}.
\]
Also,
\[
G_3 = \setlength{\arraycolsep}{9pt}
\begin{bmatrix}
    1 & 0 & 0 & 0 & 0 & 0 & 0 & 0 \\
    1 & 1 & 0 & 0 & 0 & 0 & 0 & 0 \\
    1 & 0 & 1 & 0 & 0 & 0 & 0 & 0 \\
    1 & 1 & 1 & 1 & 0 & 0 & 0 & 0 \\
    1 & 0 & 0 & 0 & 1 & 0 & 0 & 0 \\
    1 & 1 & 0 & 0 & 1 & 1 & 0 & 0 \\
    1 & 0 & 1 & 0 & 1 & 0 & 1 & 0 \\
    1 & 1 & 1 & 1 & 1 & 1 & 1 & 1 \\
\end{bmatrix}.
\]
Note that $T_8G_3 = \left(D^0_8+ D^2_8+ D^3_8+ D^5_8+ D^6_8 \right)G_3$.
Also,
\begin{equation*}
\begin{split}
    &D^0_8G_3 = G_3C_3^0,\\
    &D^2_8G_3 = G_3\sum_{\mycom{1\leq l \leq 7~\text{s.t.}}{g_{3,l+1}=1}}
    C_8^{l} = G_3C_8^2,\\
    &D^3_8G_3 = G_3\sum_{\mycom{1\leq l \leq 7~\text{s.t.}}{g_{4,l+1}=1}}
    C_8^{l} = G_3\left(C_8^1+ C_8^2+ C_8^3\right),\\
    &D^5_8G_3 = G_3\sum_{\mycom{1\leq l \leq 7~\text{s.t.}}{g_{6,l+1}=1}}
    C_8^{l} = G_3\left(C_8^1+ C_8^4+ C_8^5\right),\\
    &D^6_8G_3 = G_3\sum_{\mycom{1\leq l \leq 7~\text{s.t.}}{g_{7,l+1}=1}}
    C_8^{l} = G_3\left(C_8^2+ C_8^4+ C_8^4\right).\\
\end{split}
\end{equation*}
As a result,
$T_8G_3 = \left(D^0_8+ D^2_8+ D^3_8+ D^5_8+ D^6_8 \right)G_3 =
G_3\left(C^0_8+ C^2_8+ C^3_8+ C^5_8+ C^6_8 \right)$.
\end{example}

\begin{figure}[htbp]
\centering
	\includegraphics [width = .8\columnwidth]{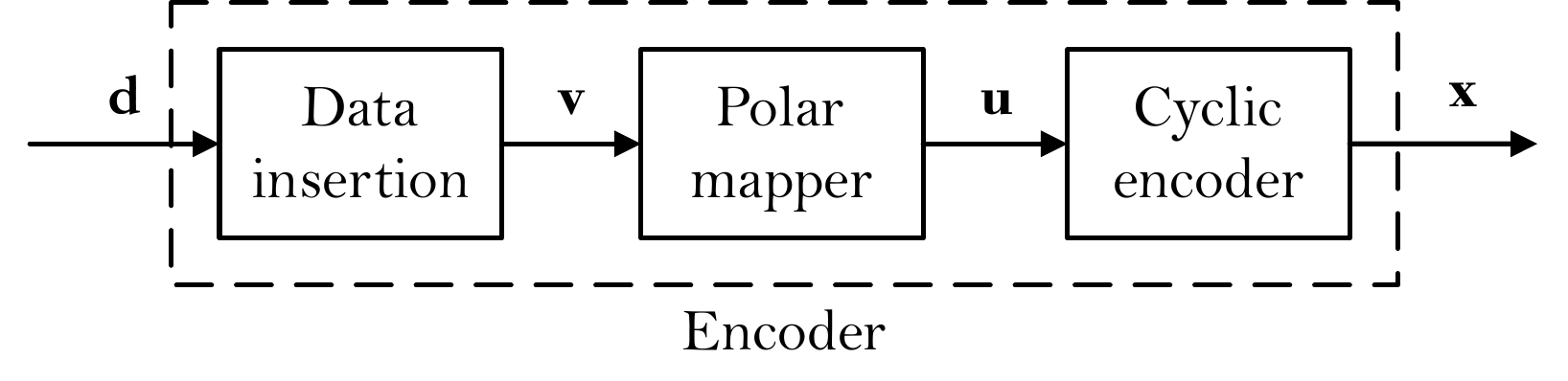}
	\caption{Proposed coding scheme equivalent to PAC codes with inner cyclic and outer polar- RM-like codes.} 
	\label{fig: PolarCyclicEncoder}
\end{figure}

Therefore, for a data word $\mathbf{v}$,
\begin{equation}\label{eq: PACv.Cyclic}
    \mathbf{v}T_NG_n = \mathbf{v}\sum_{m} D_{N}^{m}G_{n} =
    \mathbf{v}G_{n}\sum_{s} C_{N}^{s},
\end{equation}
where the second equality is by Theorem \ref{theorem}.
Fig. \ref{fig: PolarCyclicEncoder} shows the block diagram of our proposed coding scheme equivalent to the PAC coding with cyclic codes as an inner code and polar- and RM-like codes as an outer code.

Assume that the data vector $\mathbf{v}$ has its first $1$ at its $i$th location, i.e. $\mathbf{v}G_{n} = \underline{\mathbf{g}}_{i}$, where $\mathbf{g}_i$ is the $i$th row of the matrix $G_{n}$.
From \eqref{eq: PACv.Cyclic} we have
\begin{equation}
    w(\mathbf{v}T_NG_n) =
    w(\mathbf{v}G_{n}\sum_{s} C_{N}^{s}) =
    w(\underline{\mathbf{g}}_{i}\sum_{s} C_{N}^{s}) \geq w(\mathbf{g}_{i})
    ,
\end{equation}
where the inequality is by Proposition \ref{pro: 1}.
By noticing that the $d_{\text{min}}$ of polar code is the minimum row weight of $G_{N,\mathcal{A}}$,
this proves that $d_{\text{min}}$ for PAC codes is greater than or equal to $d_{\text{min}}$ for the polar- and RM-like codes.

\begin{figure}[htbp]
\centering
	\includegraphics [width = .8\columnwidth]{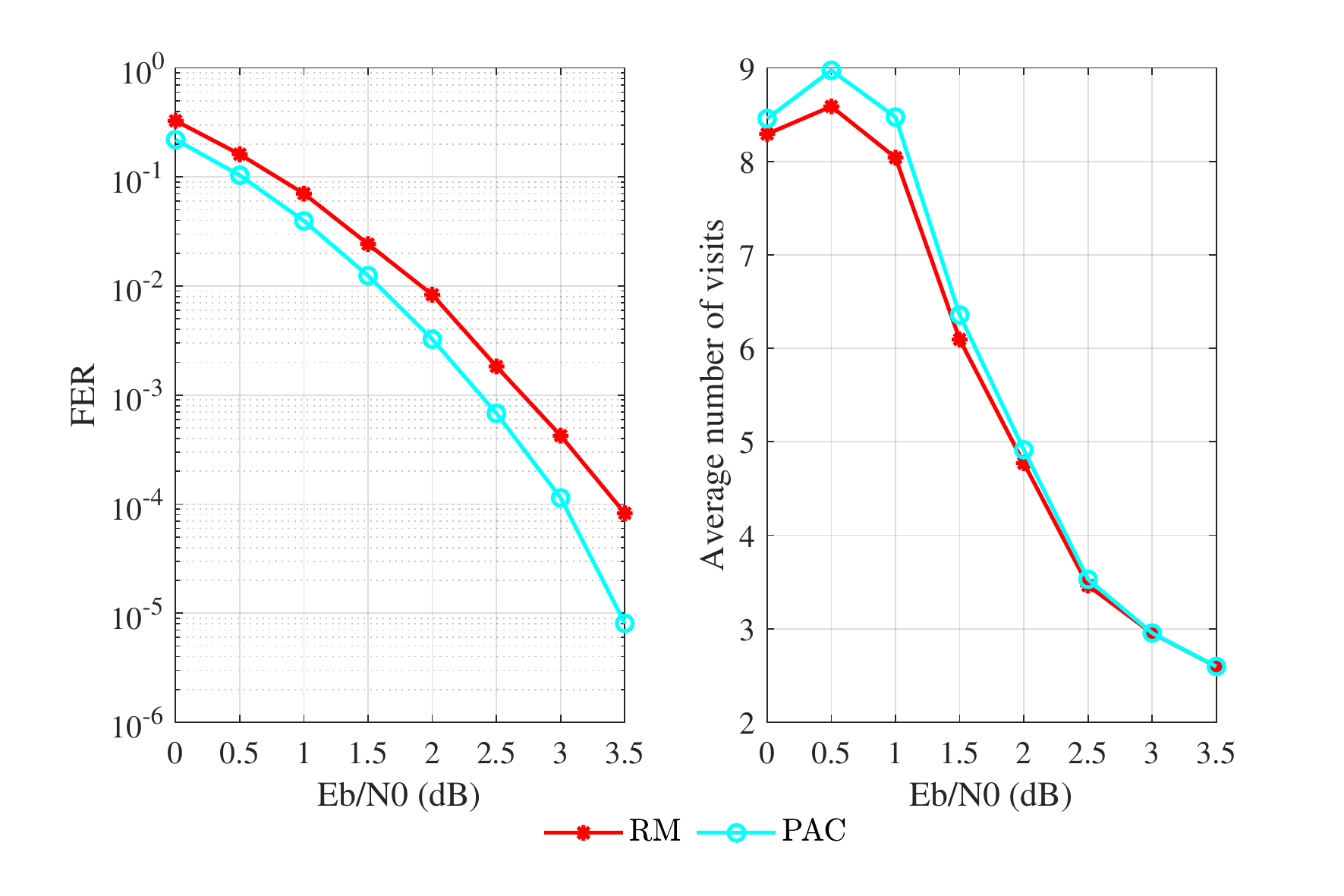}
	\caption{Performance comparison of Fano decoder of PAC and RM codes.} 
	\label{fig: PACvRM}
\end{figure}

Although this proof does not guarantee the strict improvement of the minimum distance of the PAC codes in comparison to the polar- and RM-like codes, simulation results indicate that the weight enumeration of the PAC codes can be significantly improved by selecting an appropriate convolutional code.
Fig. \ref{fig: PACvRM} compares the frame error rate (FER) and decoding complexity of PAC$(128, 29)$ and RM$(128, 29)$ codes. 
We employ the rate profile of the RM code for the PAC code as well, and the connection polynomial is $\mathbf{c}(x) = x^{10}+ x^9 + x^7 + x^3 + 1$ (3211 in octal form).
Complexity is measured by the average number of visits (ANV) per decoded bit \cite{moradi2021sequential}.
While both codes have the same complexity, this figure demonstrates that the PAC code is superior to the RM code regarding error correction performance.  
The minimum distance of both codes is $32$. However, for the RM$(128, 29)$ code, the number of codewords with a weight equal to the minimum distance is 10668, whereas, for the PAC$(128, 29)$ code, this number decreases to 324.

\section{Conclusion}\label{sec: conclusion}
Cyclic shift codes have been extensively explored in the literature, and connecting them to the PAC, RM, and polar codes offers up a field of study for enhancing the analysis of these codes.
We proved that the PAC codes are equivalent to a class of concatenated codes with cyclic codes as the inner code and polar- and RM-like codes as the outer code.
We examined the minimum distance of PAC codes with the goal of providing a new demonstration that PAC codes surpass polar- and RM-like codes in terms of weight distribution. 
We proved that adding an odd number of clockwise cyclic shifts to any row of the polar- and Reed-Muller-like codes generator matrix added with some rows below it does not reduce the row's weight. 
We used this to prove that the minimum distance for PAC codes is greater than or equal to the minimum distance for polar and RM codes. 
 

\ifCLASSOPTIONcaptionsoff
  \newpage
\fi

\bibliographystyle{IEEEtran}
\bibliography{bibliography}

\end{document}